\documentclass[11pt]{article}
\usepackage{amssymb,amsmath,amsthm}
\usepackage{tikz}
\usetikzlibrary{snakes}
\usepackage[a4paper]{geometry}

\newtheorem{theorem}{Theorem}[section]
\newtheorem{corollary}[theorem]{Corollary}
\newtheorem{lemma}[theorem]{Lemma}
\newtheorem{definition}[theorem]{Definition} \numberwithin{equation}{subsection}
\newcommand{\opseg}{\ensuremath{\mathtt{segment}}}
\newcommand{\opcol}{\ensuremath{\mathtt{color}}}
\newcommand{\drzewo}{\ensuremath{\mathcal{D}}}
\newcommand{\dom}{\ensuremath{\mathrm{dom}}}

\date{}
\author{Marek Cygan and Marcin Pilipczuk \\
\texttt{\{cygan,malcin\}@mimuw.edu.pl} \\
 \\
Department of Mathematics, Computer Science and Mechanics, \\
University of Warsaw, Warsaw, Poland\\
}

\title{Even Faster Exact Bandwidth
\thanks{
This research is partially supported by a grant from the Polish Ministry of Science and Higher
Education, project N206 005 32/0807.  }
}

\begin{document}
\maketitle

\begin{abstract}
We deal with exact algorithms for {\sc Bandwidth}, a long studied
NP-hard problem. For a long time nothing better than the trivial
$O^*(n!)$\footnote{By $O^*$ we denote standard big $O$ notation but omitting polynomial factors.}
exhaustive search was known. In 2000, Feige an Kilian~\cite{feige:exp}
came up with a $O^*(10^n)$-time algorithm. Recently we presented
algorithm that runs in $O^*(5^n)$ time and $O^*(2^n)$ space. 

In this paper we present a major modification to our algorithm which makes it run in $O(4.83^n)$ time
with the cost of $O^*(4^n)$ space complexity. This modification allowed us to perform Measure \& Conquer
analysis for the time complexity which was not used for such types of problems before.
\end{abstract}

\section{Introduction}

In this paper we focus on exact exponential-time algorithms for the {\sc Bandwidth} problem.
Let $G = (V,E)$ be an undirected graph, where $n = |V|$ and $m = |E|$. For a given one-to-one function
$\pi: V \to \{1,2\ldots,n\}$ (called {\it ordering}) its {\it bandwidth}   is the maximum difference
between positions of adjacent vertices, i.e. $\max_{uv \in E}|\pi(u)-\pi(v)|$. The {\it bandwidth} of the graph,
denoted by ${\rm bw}(G)$, is the minimum bandwidth over all orderings. The {\sc Bandwidth} problem asks
to find an ordering with bandwidth ${\rm bw}(G)$. One can consider a decision version of the {\sc Bandwidth} problem.
Precisely, we assume that the input for our problem contains
additionally an integer $b$, $1 \le b < n$. An ordering of $V$ with bandwidth at most $b$ is called a {\em
$b$-ordering}. In the decision version we ask if there exists a $b$-ordering and if that is the
case, finding it. 

A very short summary of what is known about the {\sc Bandwidth} problem follows. 
On special families of graphs ${\rm bw}(G)$ can be computed in polynomial time~\cite{assman,kleitman}. However,
in general {\sc Bandwidth} is NP-hard even on some subfamilies of trees~\cite{Garey,Monien}. Moreover
Unger~\cite{unger:bandwidth} showed that {\sc Bandwidth} problem does not belong to APX even in a
very restricted case when $G$ is a caterpillar, i.e. a very simple tree.
It is also hard for any fixed level of the W hierarchy~\cite{fellows:hardness}.
The best known polynomial-time
approximation, due to Feige \cite{feige-bandwidth-approx}, has $O(\log ^3n\sqrt{\log n \log \log n})$ approximation guarantee.
On WG 2008 we presented~\cite{naszewg} exact algorithm that runs in $O^*(5^n)$ time and $O^*(2^n)$ space.

In this paper we enhance the previous algorithm to make it run in $O(4.83^n)$ time. However, the cost of this
change is $O^*(4^n)$ space complexity, which makes this result purely theoretical. In Section \ref{s:algorytm} we
describe the enhanced algorithm. In Section \ref{s:analiza} we do Measure \& Conquer analysis
(method introduced by Fomin et al. in \cite{fgk:mis}) to obtain $O(4.83^n)$ time bound. We find this analysis
interesting, because Measure \& Conquer method at the first glance does not fit for the {\sc Bandwidth} problem at all.

\section{The algorithm}\label{s:algorytm}

The algorithm consists of two phases. First, we generate partial assignments of vertices to positions: 
we do not assign precise position to a vertex, but a segment of length $2(b+1)$ or $4(b+1)$ of possible positions. We do this in every possible way.
In the second phase for every generated segment assignment we check whether there exists a precise assignment of vertices to positions (i.e., an ordering), consistent
with the partial assignment. 

From this point we assume, that the graph $G$ has at least two vertices 
and it is connected (if $G$ is not connected we may find $b$-orderings 
of each connected component of $G$ in an independent manner).
Let us choose any (but fixed for the whole algorithm) spanning tree $\drzewo$ of the graph $G$.

\subsection{First phase: generating segment assignments}

\begin{definition}\label{def:segment}
A {\em{segment}} is a nonempty set of consecutive positions which has form of $\{i(b+1)+1, i(b+1)+2, \ldots, j(b+1)\} \cap \{1,2,\ldots, n\}$ for
some integers $i < j$. We say that this segment has index $(i, j)$ and denote it as $\Theta_{(i,j)}$.
For the sake of simplicity we define $\Theta_i = \Theta_{(i,i+1)}$ and call such segments {\em{base segments}}.
\end{definition}

\begin{definition}
A {\em{segment assignment}} is a function $\phi$ assigning a segment to every vertex such that the following conditions hold:
\begin{enumerate}
\item Every leaf of the spanning tree $\drzewo$ is assigned to a segment of size $4(b+1)$, i.e., segment $\Theta_{(i, i+4)}$ for some integer $i$.
\item Every inner vertex of $\drzewo$ is assigned to a segment of size $2(b+1)$, i.e., segment $\Theta_{(i, i+2)}$ for some integer $i$.
\item For every edge $uv$ in $\drzewo$, if vertex $u$ is the parent of the vertex $v$ and $v$ is an inner vertex, where $\phi(u) = \Theta_{(i, i+2)}$ and $\phi(v) = \Theta_{(j, j+2)}$,
  then $i$ and $j$ differ by exactly one.
\item For every leaf $v$, if $u$ is the parent of $v$ in $\drzewo$ and $\phi(v) = \Theta_{(i, i+2)}$, then $\phi(u) = \Theta_{(i-1, i+3)}$. 
\end{enumerate}
We say that a segment assignment is {\em{consistent}} with an ordering $\pi$ if for every vertex $v$ position $\pi(v)$ belongs to the segment $\phi(v)$. 
\end{definition} 
 
\begin{lemma}\label{dalekie}
Let $\pi$ be a $b$-ordering. In any segment assignment $\phi$ consistent with the ordering $\pi$, for every edge $uv$, if
$\phi(u) = \Theta_{(i, j)}$ and $\phi(v) = \Theta_{(k, l)}$ then $j \geq k$ and $l \geq i$. 
\end{lemma}

\begin{proof}
If $j < k$, then there is a gap of size at least $b+1$ between $\phi(u)$ and $\phi(v)$. Since $\pi$ is consistent with
the segment assignment, edge $uv$ is longer than $b$, contradiction. Similar argument proves that $l \geq i$.
\end{proof}

\begin{lemma}\label{istniejasegmenty}
Let $\pi$ be a $b$-ordering. There exists a segment assignment consistent with $\pi$.
\end{lemma}

\begin{proof}
As a proof we present simple construction of the segment assignment. Let $u_0$ be the root of $\drzewo$. Lets assign it to any segment of length $2(b+1)$ containing $\pi(u_0)$ (there
are exactly two possible segments). Then assign segments to vertices in the root-to-leaf order. Let $u$ be an unassigned vertex with parent $v$ and let 
$T_{(i, i+2)}$ be the segment assigned to $v$. Since $i(b+1) + 1 \leq \pi(v) \leq (i+2)(b+1)$ and $\pi$ is a $b$-ordering,
$$\pi(u) \in \Theta_{(i-1, i+3)} = \Theta_{(i-1, i+1)} \cup \Theta_{(i+1, i+3)}.$$
We can assign $v$ to $\Theta_{(i-1, i+3)}$ or to one of the segments $\Theta_{(i-1, i+1)}$ and $\Theta_{(i+1, i+3)}$, containing $\pi(u)$, depending
whether $v$ is an inner vertex or a leaf.
\end{proof}

Our goal in the first phase is to generate a set of segment assignments such that for every $b$-ordering there exists generated segment assignment
consistent with it. In the second phase we check for every segment assignment whether consistent $b$-ordering exists. As a result we
check if there exists any $b$-ordering of the given graph $G$. 

The first phase is as follows:
\begin{enumerate}
\item Assign root $u_0$ of the tree $\drzewo$ to any valid segment of size $2(b+1)$.
\item Recursively assign other vertices in the root-to-leaf order in tree $\drzewo$. Given vertex $u$ with parent $v$ assigned to segment $\Theta_{(i, i+2)}$ assign $u$
to one of the segments $\Theta_{(i-1, i+1)}$, $\Theta_{(i+1, i+3)}$ if $u$ is an inner vertex, or to $\Theta_{(i-1, i+3)}$ when $u$ is a leaf (see Figure~\ref{rys:podzialy_podwojne}).\label{alg:rek}
\item For every generated assignment check condition from Lemma \ref{dalekie}: in other words accept assignment iff for every edge $uv$ in $G$ segments
assigned to vertices $u$ and $v$ are not too far from each other.
\end{enumerate}

\begin{figure}[htp]
\begin{center}
\begin{tikzpicture}
    
    \definecolor{gray1}{rgb}{0.7,0.7,0.7}
    \definecolor{gray2}{rgb}{0.3,0.3,0.3}
    \fill[gray1] (0,0) rectangle (6,0.5);
    \fill[gray2] (6,0) rectangle (12,0.5);
    \draw[very thin,step=.5cm] (0,0) grid (6,0.5);
    \draw[very thin,step=.5cm] (6,0) grid (12,0.5);
    \draw[very thin,step=.5cm] (3,1.5) grid (9,2);
    \draw[very thin] (3,1.5) -- (3,2);
    \draw[very thin] (3,1.5) -- (9,1.5);

    \foreach \x in {3cm, 6cm, 9cm}
    {
      \draw[very thick] (\x,1.2cm) -- (\x,2.3cm);
    }

    \foreach \x in {0cm, 3cm, 6cm, 9cm, 12cm}
    {
      \draw[very thick] (\x,-0.3cm) -- (\x,0.8cm);
    }
    \draw[thin, dashed,<->,stealth-stealth] (3,2.3) -- (6,2.3);
    \draw[thin, dashed,<->,stealth-stealth] (6,2.3) -- (9,2.3);
    \draw (4.5, 2.5) node {$b+1$};
    \draw (7.5, 2.5) node {$b+1$};

    \draw (2.75, 1.75) node {$...$};
    \draw (9.25, 1.75) node {$...$};
    \draw (-0.25, 0.25) node {$...$};
    \draw (12.25, 0.25) node {$...$};
\end{tikzpicture}
\caption{Upper part of the picture shows the segment $\Theta_{i,i+2}$ assigned to the vertex $v$. Lower part contains
  possible positions for the vertex $u$ (child of $v$ in the tree $D$) covered by two segments $\Theta_{i-1,i+1}$ and $\Theta_{i+1,i+3}$.}
\label{rys:podzialy_podwojne}
\end{center}
\end{figure}
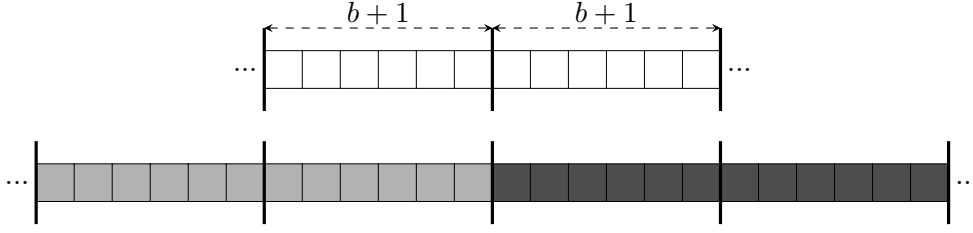

Note that Lemma \ref{istniejasegmenty} implies that for every $b$-ordering $\pi$ there is a generated segment assignment consistent with $\pi$. 
Note that this is exactly the assignment described in the proof of Lemma \ref{istniejasegmenty}.

The second phase is fed with generated segment assignments on-line. Therefore the first phase uses polynomial space. 
In Step \ref{alg:rek} of the algorithm we have two possibilities for every inner non-root vertex. For the root we have
$\lceil\frac{n}{b+1}\rceil + 1 \leq n+1$ possibilities (additional $1$ comes from segment $\Theta_{-1,1}$, which is also
correct segment according to the Definition~\ref{def:segment}). We do not have a choice for leaves, 
thus the algorithm generates at most $(n+1)2^{n-1}$ assignments. This algorithm uses polynomial time
for every generated assignment, so this phase works in $O^*(2^n)$ time.

\subsection{Second phase: depth-first search}

In this phase for every generated segment assignment we check if there exists a $b$-ordering consistent with the segment assignment. Let us denote this
segment assignment by $\phi$.
This phase is very similar to the second phase of our $O^*(5^n)$ algorithm \cite{naszewg}, 
but adapted to the segments of size $2(b+1)$ and $4(b+1)$. This difference allows us to use Measure \& Conquer method in complexity analysis.
First, we recall some key definitions and facts from
the $O^*(5^n)$ algorithm.

We assign a vertex to each position one by one, but
the main idea is the order in which we fill those positions in.
For every position $i$, let
$\opseg(i) = \lceil \frac{i}{b+1} \rceil$ be the base segment number
of this position, and let $\opcol(i) = ((i-1) \mod (b+1)) + 1$ be
the index of the position in its base segment, which we call the {\em{color}}
of this position. Note that the color of position is the remainder of this
number modulo $b+1$, but in the range $[1,b+1]$ instead of $[0,b]$.  

Let us sort positions lexicographically according to pairs $(\opcol(i), \opseg(i))$.
To each of those positions we assign a vertex, in exactly this order. We 
call this ordering {\em{the color order}} of positions.

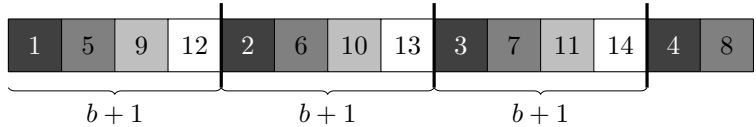
\begin{figure}[htbp]
\begin{center}
{\small{
\begin{tikzpicture}[scale=0.7]

    \foreach \x in {1, ..., 4}
    {
      \fill[darkgray] (\x * 4.0 - 4.0, 0.0) rectangle (\x * 4.0 - 3.0, 1.0);
      \draw[white] (\x * 4.0 - 3.5, 0.5) node {$\x$};
    }
    \foreach \x in {5, ..., 8}
    {
      \fill[gray] (\x * 4.0 - 5.0 * 4.0 + 1.0, 0.0) rectangle (\x * 4.0 - 5.0 * 4.0 + 2.0, 1.0);
      \draw (\x * 4.0 - 5.0 * 4.0 + 1.5, 0.5) node {$\x$};
    }
    \foreach \x in {9, 10, 11}
    {
      \fill[lightgray] (\x * 4.0 - 9.0 * 4.0 + 2.0, 0.0) rectangle (\x * 4.0 - 9.0 * 4.0 + 3.0, 1.0);
      \draw (\x * 4.0 - 9.0 * 4.0 + 2.5, 0.5) node {$\x$};
    }
    \foreach \x in {12, 13, 14}
    {
      \fill[white] (\x * 4.0 - 12.0 * 4.0 + 3.0, 0.0) rectangle (\x * 4.0 - 12.0 * 4.0 + 4.0, 1.0);
      \draw (\x * 4.0 - 12.0 * 4.0 + 3.5, 0.5) node {$\x$};
    }
  	\draw[step=1cm] (0,0) grid (14,1);
    \foreach \x in {1,2,3}
    {
      \draw[very thick] (\x * 4.0, -0.3) -- (\x * 4.0, 1.3);
      \draw[very thin, snake=brace, mirror snake] (\x * 4.0 - 4.0, -0.3) -- (\x * 4.0, -0.3);
      \draw (\x * 4.0 - 2.0, -0.8) node {$b+1$};
    }

\end{tikzpicture}
}}
\caption{Color order of positions for $n=14$ and $b=3$.}
\label{pierwszyrysunek}
\end{center}
\end{figure}

Following lemma is the key observation in our algorithm.

\begin{lemma}\label{lemat-najwazniejsze}
Ordering $\pi$
is a $b$-ordering iff for every edge $uv$ such that $\opseg(\pi(u)) < \opseg(\pi(v))$ we have $\opseg(\pi(u)) + 1 = \opseg(\pi(v))$ and $\opcol(\pi(u)) > \opcol(\pi(v))$.
\end{lemma}

\begin{proof}
Since $\pi$ is a $b$-ordering, for every edge $uv$ we have $|\opseg(\pi(u)) - \opseg(\pi(v))| \leq 1$. 
If $\opseg(\pi(u)) = \opseg(\pi(v))$ then $uv$ is not longer than $b$. Otherwise, suppose w.l.o.g. that 
$\opseg(\pi(u)) + 1 = \opseg(\pi(v))$. Note that the distance between positions with the same color
in the neighboring segments is exactly $b+1$, so $uv$ is not longer than $b$ iff $u$ has greater color than $v$ (see Figure \ref{drugirysunek}).
\end{proof}

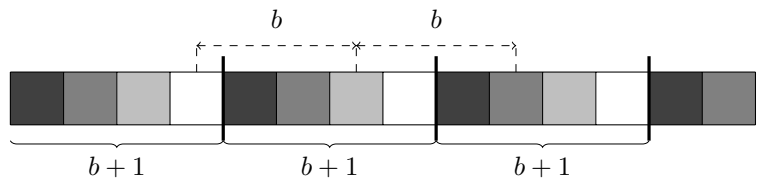
\begin{figure}[htbp]
\begin{center}
{\small{
\begin{tikzpicture}[scale=0.7]

    \foreach \x in {0, ..., 3}
    {
      \fill[darkgray] (\x * 4.0 + 0.0, 0.0) rectangle (\x * 4.0 + 1.0, 1.0);
    }
    \foreach \x in {4, ..., 7}
    {
      \fill[gray] (\x * 4.0 - 4.0 * 4.0 + 1.0, 0.0) rectangle (\x * 4.0 - 4.0 * 4.0 + 2.0, 1.0);
    }
    \foreach \x in {8, 9, 10}
    {
      \fill[lightgray] (\x * 4.0 - 8.0 * 4.0 + 2.0, 0.0) rectangle (\x * 4.0 - 8.0 * 4.0 + 3.0, 1.0);
    }
    \foreach \x in {11, 12, 13}
    {
      \fill[white] (\x * 4.0 - 11.0 * 4.0 + 3.0, 0.0) rectangle (\x * 4.0 - 11.0 * 4.0 + 4.0, 1.0);
    }
  	\draw[step=1cm] (0,0) grid (14,1);
    \foreach \x in {1,2,3}
    {
      \draw[very thick] (\x * 4.0, -0.3) -- (\x * 4.0, 1.3);
      \draw[very thin, snake=brace, mirror snake] (\x * 4.0 - 4.0, -0.3) -- (\x * 4.0, -0.3);
      \draw (\x * 4.0 - 2.0, -0.8) node {$b+1$};
    }
    \draw[dashed, very thin] (3.5,1.0) -- (3.5,1.5);
    \draw[dashed, very thin] (6.5,1.0) -- (6.5,1.5);
    \draw[stealth-stealth, <->, very thin, dashed] (3.5, 1.5) -- (6.5, 1.5);
    \draw (5.0, 2) node {$b$};
    \draw[dashed, very thin] (9.5,1.0) -- (9.5,1.5);
    \draw[stealth-stealth, <->, very thin, dashed] (6.5, 1.5) -- (9.5, 1.5);
    \draw (8.0, 2) node {$b$};
\end{tikzpicture}
}}
\caption{Picturable proof of the Lemma \ref{lemat-najwazniejsze}\label{drugirysunek}}
\end{center}
\end{figure}

\begin{corollary}\label{najwazniejsze}
Ordering $\pi$, is a $b$-ordering iff for every edge $uv$ satisfying
$\opseg(\pi(u)) + 1 = \opseg(\pi(v))$ vertex
$u$ is assigned to greater position in the color order than vertex $v$ and for every edge
$uv$ values $\opseg(\pi(u))$ and $\opseg(\pi(v))$ differ by at most one.
\end{corollary}

\begin{definition}
By {\em{state candidate}} we denote partial function $s:V \to \{\Theta_i: 0 \leq i < \lceil \frac{n}{b+1} \rceil\}$, i.e., 
   assigning base segments to some of the vertices, such that for every vertex $v$ if $s(v)$ is defined then $s(v) \subset \phi(v)$ (where $\phi$ is the fixed assignment for this phase).
\end{definition}

Note, that for an inner vertex $v$ of the tree $\drzewo$ there are two possible values for the state candidate $s$ and there are four possible values
for $v$ being a leaf.

\begin{definition}
A {\em{state}} is a state candidate $s$ satisfying the following conditions:
\begin{enumerate}
\item Vertices from $\dom(s)$ can be assigned to the first $|\dom(s)|$ positions in the color order,
  in a way compatible with $s$.\label{state:1}
\item For every edge $uv$ one of the following holds.\label{state:2}
  \begin{enumerate}
    \item Both $s(u)$ and $s(v)$ are undefined.
    \item Exactly one value among $\{s(u), s(v)\}$ is defined; say $s(v) = \Theta_i$ and $s(u)$ is undefined. Then if $\phi(u) = \Theta_{(k, l)}$ we have $k \leq i$.
    \item Both $s(u)$ and $s(v)$ are defined; say $s(v) = \Theta_i$ and $s(v) = \Theta_k$. Then $|i-k| \le 1$. \label{state:2c}
  \end{enumerate}
\end{enumerate}
\end{definition}

\begin{definition}
We say that a state $s'$ is an {\em{extension}} of some state $s$ if there exists vertex $v$ such that:
\begin{enumerate}
\item $s(v)$ is undefined and $s'(v)$ is defined;\label{extension:1}
\item $\dom(s') = \dom(s) \cup \{v\}$ and $s = s'|_{\dom(s)}$;\label{extension:2}
\item If $uv$ is an edge in the {\bf whole graph $G$} and $s'(u) = \Theta_k$ and $s'(v) = \Theta_i$, then $k-1 \leq i \leq k$.\label{extension:3}
\end{enumerate}
\end{definition}

\begin{lemma}\label{l:iff1}
Let $\pi$ be a $b$-ordering compatible with the given segment assignment $\phi$. By $s_k$ for $0 \leq k \leq n$ we denote the state candidate, assigning
$\opseg(\pi(v))$ to every vertex $v$ assigned to one of the first $k$ positions in the color order (i.e., $|\dom(s_k)| = k$). Then every $s_k$ is a state
and for every $0 \leq k < n$ state $s_{k+1}$ is an extension of state $s_k$.
\end{lemma}

\begin{proof}
This is a straightforward corollary from Corollary \ref{najwazniejsze}. Point \ref{state:1} of the state definition is obviously satisfied by the definition of $s_k$.
Now lets look at any $s_k$ and edge $uv$. If both $s_k(u)$ and $s_k(v)$ are undefined, Point \ref{state:2} of the state definition is satisfied.
If both are defined, Corollary \ref{najwazniejsze} clearly implies Point \ref{state:2}. Assume then that $s_k(v) = \Theta_i$ and $s_k(u)$ is undefined. By construction
of state candidates $s_k$, $\pi(u)$ is later in the color order than $\pi(v)$. Therefore, by Corollary \ref{najwazniejsze}, $\opseg(\pi(u)) \leq \opseg(\pi(v))$. Since
$\pi(u) \in \opseg(\pi(u)) \subset \phi(u)$, it implies the condition in Point \ref{state:2} of the state definition.

Let us now prove that  $s_{k+1}$ is an extension of the state $s_k$. Points \ref{extension:1} and \ref{extension:2} of the extension definition follow directly from the construction of $s_k$.
Note that if $v$ is the vertex defined by $s_{k+1}$ but not by $s_k$, then $\pi(v)$ is later in the color order than any $u$ defined by $s_k$. 
Therefore, if we take any $uv$ as in Point \ref{extension:3} of the extension definition, by Corollary \ref{najwazniejsze} the condition in Point \ref{extension:3} is satisfied.
\end{proof}

\begin{lemma}\label{l:iff2}
Let $s_k$ for $0 \leq k \leq n$ be a set of states such that $s_{k+1}$ is an extension of state $s_k$ for every $0 \leq k < n$. Let $v_k$ be the vertex defined
by $s_k$, but not by $s_{k-1}$. Then ordering $\pi$ which assigns position $k$ in the color order to the vertex $v_k$ is a $b$-ordering.
\end{lemma}

\begin{proof}
We use Corollary \ref{najwazniejsze} once again. Let us prove the thesis by contradiction. Assume that edge $uv$ is longer than $b$ and w.l.o.g. $\pi(u) < \pi(v)$.
Since $\pi(u) + b < \pi(v)$, $\opseg(\pi(u)) < \opseg(\pi(v))$. If $\opseg(\pi(u)) + 1 < \opseg(\pi(v))$, then $s_n$ is not a state due to Point \ref{state:2c} of
the state definition. Therefore $\opseg(\pi(u)) + 1 = \opseg(\pi(v))$ and, by Corollary \ref{najwazniejsze}, $\pi(u)$ is earlier in the color order than $\pi(v)$.
But then there exists $s_k$ for which $s_k(u)$ is defined, $s_k(v)$ is undefined but $s_{k+1}(v)$ is defined. Thus we have contradiction 
with Point \ref{extension:3} of the extension definition for the edge $uv$.
\end{proof}

The algorithm for second phase is now quite clear. Given segment assignment $\phi$, we do depth-first search over states, looking for sequence
of states and extensions as in Lemmas \ref{l:iff1} and \ref{l:iff2}. More precisely:
\begin{enumerate}
\item Start with a state $\emptyset$.
\item Having state $s$ with $|\dom(s)| = k < n$, try to create extension $s'$ of $s$ in every possible way, i.e., try to assign every undefined vertex in $s$
to the base segment, where $(k+1)$-th position in the color order lies.
\item If we reach state $s$ with $|\dom(s)| = n$, construct the $b$-ordering $\pi$ using the DFS stack (that is, states $\emptyset = s_0$, $s_1$, \ldots, $s_n=s$) and return it.\label{alg2:true}
\item If no state $s$ with $|\dom(s)| = n$ is reached, return that there is no $b$-ordering compatible with $\phi$.\label{alg2:false}
\end{enumerate}

Note that Lemmas \ref{l:iff1} and \ref{l:iff2} imply that this algorithm returns correct $b$-ordering in Step \ref{alg2:true} and there are no correct $b$-orderings
if the algorithm reaches Step \ref{alg2:false}. Therefore we proved that this algorithm is correct. In Section \ref{s:analiza} we prove the $O(4.83^n)$ time bound and $O^*(4^n)$ space bound.

\section{Time and space analysis}\label{s:analiza}

\subsection{Memory bound}

This analysis is fairly easy. Note that only non-polynomial space used in the algorithm is the space used to track visited states in the depth-first search.
We try to bound the number of visited states for one fixed segment assignment $\phi$ by $4^n$. This implies $O^*(4^n)$ space bound.

\begin{lemma}\label{lemma:mem}
In one run of the second phase, the algorithm visits at most $3^{n-L} 4^L \leq 4^n$ states, where $L$ is the number of leaves in the tree $\drzewo$.
\end{lemma}

\begin{proof}
Let $s$ be a visited state and $v$ be a vertex.

If $v$ is an inner vertex of $\drzewo$ with $\phi(v) = \Theta_{(i, i+2)}$, then $s(v)$ is either undefined or $s(v) = \Theta_i$ or $s(v) = \Theta_{i+1}$ - three possibilities.

If $v$ is a leaf of $\drzewo$ with $\phi(v) = \Theta_{(i-1, i+3)}$, then $s(v)$ is either undefined or one of the base segments $\Theta_{i-1}$, $\Theta_i$, $\Theta_{i+1}$, $\Theta_{i+2}$ - five possibilities, too many.
However, lets look at the parent $u$ of the vertex $v$ in the tree $\drzewo$. By construction, $\phi(u) = \Theta_{(i, i+2)}$. Note that, by the definition of the state, since
$uv$ is an edge of $G$:
\begin{enumerate}
\item If $s(u)$ is undefined, then $s(v)$ cannot be $\Theta_{i-1}$.
\item If $s(u) = \Theta_i$, then $s(v)$ cannot be $\Theta_{i+2}$.
\item If $s(u) = \Theta_{i+1}$, then $s(v)$ cannot be $\Theta_{i-1}$.
\end{enumerate}
Therefore, in every case, there are only four possibilities for a leaf $v$. This leads to $3^{n-L}4^L$ bound for the number of valid states.
\end{proof}

\subsection{Time bound}

Note that search for possible expansions of given state and checking if a state candidate is a state can be done in polynomial time. 
Therefore the time used by the whole algorithm is bounded by $O((s(G) + 2^n)n^\gamma)$ where $\gamma$ is a constant and $s(G)$ is the number of visited states.
The $2^n$ factor is due to the first phase of the algorithm. Now we focus on bounding $s(G)$, i.e., the number of visited states in the run of the whole algorithm.

\begin{lemma}\label{lemma:time}
The algorithm in the whole run visits at most $3(n+1)\kappa^n$ states for some constant $\kappa < 4.83$.
\end{lemma}

\begin{proof}
Let us do indepth analysis of the number of possible states similar, but more broad, to that in the proof of Lemma \ref{lemma:mem}.
Let $s$ be a state visited while considering segment assignment $\phi$.

Let $u_0$ be the root of $\drzewo$. Then $\phi(u_0)$ has at most $\lceil \frac{n}{b+1} \rceil + 1$ possible values. 
If $\phi(u_0) = \Theta_{(i, i+2)}$, then $s(u_0)$ is either undefined or equal to $\Theta_i$ or $\Theta_{i+1}$. 
In total, $3(n+1)$ possibilities for $u_0$.

Let $u$ be a parent of a leaf $v$. Then there are at most four possibilities to choose for $\phi(v)$ and $s(v)$, 
knowing $\phi(u)$ and $s(u)$. The analysis is the same as in the proof of Lemma \ref{lemma:mem}.

Let now $v$ be an inner vertex with parent $u$. Let $\phi(u) = \Theta_{(i, i+2)}$. Then 
$\phi(v)$ is either $\Theta_{(i-1, i+1)}$ or $\Theta_{(i+1, i+3)}$, by the way we construct $\phi$.
The following restrictions are implied by the fact that $uv$ is an edge of $G$ and
by the state definition.

\renewcommand{\theenumi}{\Alph{enumi}}
\renewcommand{\labelenumi}{\Alph{enumi}.}

\begin{enumerate}
\item If $s(u)$ is undefined:\label{a:opt:u}
  \begin{enumerate}
    \item If $\phi(v) = \Theta_{(i-1, i+1)}$, then $s(v)$ is undefined or $s(v) = \Theta_i$. State definition forbids possibility that $s(v) = \Theta_{i-1}$.
    \item If $\phi(v) = \Theta_{(i+1, i+3)}$, then $s(v)$ is undefined, $s(v) = \Theta_{i+1}$ or $s(v) = \Theta_{i+2}$. 
  \end{enumerate}
  In total, {\bf{five}} possibilities for $\phi(v)$ and $s(v)$.
\item If $s(u) = \Theta_i$:\label{a:opt:l}
  \begin{enumerate}
    \item If $\phi(v) = \Theta_{(i-1, i+1)}$, then $s(v)$ is undefined, $s(v) = \Theta_{i-1}$ or $s(v) = \Theta_i$.
    \item If $\phi(v) = \Theta_{(i+1, i+3)}$, then $s(v) = \Theta_{i+1}$. State definitions forbids possibilities that $s(v)$ is undefined and that $s(v) = \Theta_{i+2}$.
  \end{enumerate}
  In total, {\bf{four}} possibilities for $\phi(v)$ and $s(v)$.
\item If $s(u) = \Theta_{i+1}$:\label{a:opt:r}
  \begin{enumerate}
    \item If $\phi(v) = \Theta_{(i-1, i+1)}$, then $s(v)$ is undefined or $s(v) = \Theta_i$. State definition forbids possibility that $s(v) = \Theta_{i-1}$. 
    \item If $\phi(v) = \Theta_{(i+1, i+3)}$, then $s(v)$ is undefined, $s(v) = \Theta_{i+1}$ or $s(v) = \Theta_{i+2}$.
  \end{enumerate}
  In total, {\bf{five}} possibilities for $\phi(v)$ and $s(v)$.
\end{enumerate}

In every option we got at most five options for $\phi(v)$ and $s(v)$ values for every inner, non-root vertex $v$. Together with
four possibilities for leaves and $3(n+1)$ for root, this proves that the algorithm visits at most $3(n+1)5^n$ states.

However, there are few places where we have four, not five possibilities for $\phi(v)$ and $s(v)$: when $v$ is a leaf or when
the parent of $u$ is assigned by the state to the left (smaller position numbers) half of its segment (Option \ref{a:opt:l}). Moreover, in every moment, when we
have five possibilities for $\phi(v)$ and $s(v)$, vertex $v$ might be assigned to the left (smaller position numbers) half of its segment, which gives us Option \ref{a:opt:l}
for the analysis of children of $v$. This leads us to the conclusion that we can use Measure \&{} Conquer technique to obtain better bound.

The Measure \&{} Conquer method was introduced by Fomin, Grandoni and Kratsch (see \cite{fgk:mis}). 
As in the above analysis, we analyze vertices in the root-to-leaf order. We use Measure \&{} Conquer method to estimate
number of possible states.

Let us consider hypothetical state generator that generates possible states (i.e., pairs of functions $\phi$ and $s$), by analyzing
the tree $\drzewo$ in the root-to-leaf order. The generator first sets $\phi(u_0)$ and $s(u_0)$ in every possible way (at most $3(n+1)$ ways).
Then, while analyzing vertex $v$ with already set $\phi(u)$ and $s(u)$ for parent $u$ of $v$, it assigns $\phi(v)$ and $s(v)$ in every
possible way, keeping in mind limitations described above, both for $v$ being a leaf and inner vertex. We use Measure \&{} Conquer to estimate
number of generated states by this generator.

At any step of the generator we measure the weight of the current problem instance, i.e., already constructed functions $\phi$ and $s$.
The weight of the instance is the sum of weights of vertices. Let $\alpha, \beta \in [0,1]$ be constants to be defined lated. 
The weight of the vertex $v$ is:

\renewcommand{\theenumi}{\Roman{enumi}}
\renewcommand{\labelenumi}{\Roman{enumi}.}
\begin{enumerate}
\item $0$, if $v$ is already analyzed;
\item $1$, if $v$ is not analyzed and parent of $v$ is not analyzed or $v$ is the root of $\drzewo$;
\item $1$, if $v$ is not analyzed, parent $u$ of $v$ is analyzed and $s(u)$ is undefined;\label{x:opt:1}
\item $\alpha$, if $v$ is not analyzed, parent $u$ of $v$ is analyzed and $\phi(u) = \Theta_{(i, i+2)}$ and $s(u) = \Theta_i$ for some integer $i$.\label{x:opt:2}
\item $\beta$, if $v$ is not analyzed, parent $u$ of $v$ is analyzed and $\phi(u) = \Theta_{(i, i+2)}$ and $s(u) = \Theta_{i+1}$ for some integer $i$.\label{x:opt:3}
\end{enumerate}

Let $T(w)$ be a bound for number of states generated by the generator from point, where the size of the instance is at most $w$. Now we estimate $T$, 
using aforementioned limitations for generator choices.

Let $v$ be a non-root vertex currently analyzed by the generator, where vertex $u$ is its parent.
If $v$ is a leaf, there are always four possibilities for $\phi(v)$ and $s(v)$, independent of whether $v$ falls into Category \ref{x:opt:1}, \ref{x:opt:2}, \ref{x:opt:3}, i.e.,
whether $v$ weights $1$, $\alpha$ or $\beta$. Therefore
$$T(w) \leq \max (4T(w-1), 4T(w-\alpha), 4T(w-\beta)).$$

Now lets look at the case when $v$ is an inner vertex. Let $\phi(u) = \Theta_{(i, i+2)}$.
\renewcommand{\theenumi}{\arabic{enumi}}
\renewcommand{\labelenumi}{\arabic{enumi}.}
\begin{enumerate}
\item If $s(u)$ is undefined, $v$ has got weight $1$ and we have five possibilities for $\phi(v)$ and $s(v)$. In two of them, $v$ falls under Category \ref{x:opt:1} for children of $v$,
  in two --- under Category \ref{x:opt:3}, and in one --- under Category \ref{x:opt:2}. Vertex $v$ becomes analyzed and has weight $0$. Since there is at least one child of vertex $v$,
  the following bound holds in this case:
  $$T(w) \leq 2T(w-1) + 2T(w - 1 - (1 - \beta)) + T(w-1-(1-\alpha)).$$
\item If $s(u) = \Theta_i$, $v$ has got weight $\alpha$ and we have four possibilities for $\phi(v)$ and $s(v)$. In one of them, $v$ falls under Category \ref{x:opt:1} for children of $v$,
  in one --- under Category \ref{x:opt:3}, and in two --- under Category \ref{x:opt:2}. Vertex $v$ becomes analyzed and has weight $0$. Since there is at least one child of vertex $v$,
  the following bound holds in this case:
  $$T(w) \leq T(w-\alpha) + T(w - \alpha - (1-\beta)) + 2T(w - \alpha - (1-\alpha)).$$
\item If $s(u) = \Theta_{i+1}$, $v$ has got weight $\beta$ and we have five possibilities for $\phi(v)$ and $s(v)$. In two of them, $v$ falls under Category \ref{x:opt:1} for children of $v$,
  in two --- under Category \ref{x:opt:3}, and in one --- under Category \ref{x:opt:2}. Vertex $v$ becomes analyzed and has weight $0$. Since there is at least one child of vertex $v$,
  the following bound holds in this case:
  $$T(w) \leq 2T(w-\beta) + 2T(w - \beta - (1 - \beta)) + T(w-\beta-(1-\alpha)).$$
\end{enumerate}

By searching the space of possible values $\alpha$ and $\beta$ and by solving the above equations numerically,
we got that for $\alpha = 0.8805$ and $\beta = 1$ function $T(w)$ is bounded by $\kappa^n$ for $\kappa \sim 4.828485 < 4.83$.
This completes the proof.
\end{proof}

The following theorem is a straightforward corollary from Lemmas \ref{lemma:mem} and \ref{lemma:time}.
\begin{theorem}
There exists an algorithm that solves the {\sc{Bandwidth}} problem in $O(4.83^n)$ time and $O^*(4^n)$ space.
\end{theorem}

\bibliographystyle{plain}
\bibliography{bandwidth-mc}

\end{document}